\documentclass{amsart}
\usepackage{amssymb,comment}

\newcommand{\bv}{\mathcal{BV}}
\newcommand{\cobar}{\Omega}
\newcommand{\del}{\partial}
\newcommand{\hyc}{\mathcal{H}y}
\newcommand{\id}{\mathrm{id}}
\newcommand{\Oo}{\mathcal{O}}
\newcommand{\ash}{{\scriptstyle \text{\rm !`}}}
\newcommand{\Star}{\bigstar}
\newcommand{\sS}{\mathbb{S}}
\newcommand{\wtm}{\sqcup}

\theoremstyle{plain} 
\newtheorem{lemma}{Lemma}
\newtheorem*{mainthm}{Main Theorem}
\newtheorem{proposition}[lemma]{Proposition}
\newtheorem{corollary}[lemma]{Corollary}
\theoremstyle{definition}
\newtheorem*{question*}{Question}
\newtheorem{definition}[lemma]{Definition}
\newtheorem*{construction}{Construction}

\newtheorem*{remark}{Remark}

\title[Formal formality of low dimensional Calabi-Yaus]{Formal formality of the hypercommutative algebras of low dimensional Calabi-Yau varieties}
\author{Gabriel C. Drummond-Cole}
\thanks{This material is based upon work supported by the National Science Foundation under Award No. DMS-1004625.}
\begin{document}
\maketitle

\begin{abstract}
There is a homotopy hypercommutative algebra structure on the cohomology of a Calabi-Yau variety.  The truncation of this homotopy hypercommutative algebra to a strict hypercommutative algebra is well-known as a mathematical realization of the genus zero B-model.  It is shown that this truncation loses no information for some cases, including all Calabi-Yau $3$-folds.
\end{abstract}
\section{Introduction}
In~\cite{BarannikovKontsevich:FMFLAPV}, Barannikov and Kontsevich gave a mathematical construction of the genus zero B-model with a Calabi-Yau variety $M$ of arbitrary dimension as the target.\footnote{By Calabi-Yau variety, we shall always mean a smooth compact K\"ahler manifold with $H^1=0$ and a nowhere vanishing holomorphic volume form.}
They started by considering a linear space $B$, namely the Dolbeaut resolution of the sheaf of holomorphic polyvector fields.  Using the various geometric structures present on $M$, they gave $B$ the structure of a differential graded $\bv$-algebra.  Next, they showed how to use this $\bv$ structure to induce the structure of a hypercommutative algebra, or $\hyc$-algebra, on the cohomology of $B$.  

In~\cite{DrummondColeVallette:MMBVO}, Bruno Vallette and the author investigated the Barannikov-Kontsevich passage from a $\bv$-algebra to a $\hyc$-algebra in a homotopical, operad-theoretic context.  There, it was shown that the output of the Barannikov-Kontsevich construction is properly thought of as the first piece of a {\em homotopy} $\hyc$-algebra structure on the cohomology of $B$ which can be truncated to the strict $\hyc$-algebra structure of Barannikov-Kontsevich.  One can easily construct examples where this truncation loses information.  That is, in such an example, the higher pieces of this homotopy $\hyc$-algebra structure contain information that is not encapsulated in the truncated strict $\hyc$-algebra structure.  As these examples do not arise from Calabi-Yau varieties, it is natural to ask:
\begin{question*}
Let $B$ be the Barannikov-Kontsevich $\bv$-algebra constructed from a Calabi-Yau variety. Is the homotopy $\hyc$-algebra structure on the complex-valued cohomology of $B$ formal?  In other words, is the structure on the cohomology of $B$ isomorphic to one where all the higher homotopy operations vanish?
\end{question*}
The purpose of this paper is to give an affirmative answer in the case of all Calabi-Yau varieties in dimensions two and three and the class of projective hypersurfaces in dimension four.  Specifically, the theorem is as follows:
\begin{mainthm}
Let $M$ be any Calabi-Yau variety of dimension two or three, or a Calabi-Yau variety of dimension four which has nonzero Hodge numbers only in bidegrees where a hypersurface of dimension four has nonzero Hodge numbers. Then the induced homotopy $\hyc$ structure on the cohomology of the associated $\bv$ algebra $B$ is formal.  
\end{mainthm}

The paper is organized as follows: the next two sections describe how to transfer a $\bv$ structure satisfying special conditions on a complex $B$ to a homotopy $\hyc$ structure on the cohomology of $B$ given appropriate transfer data.  Section~\ref{sec:CY} reviews the Calabi-Yau example and describes the transfer data.  Section~\ref{sec:topandbottom} uses a universal geometric argument to show that the product is the only transfered operation that can land in the top degree, and the final section uses degree-counting arguments to show that there are no higher operations in low dimension.  It is in this sense that the formality of the hypercommutative algebra is a {\em formal} formality: it does not depend on intricate analysis of the geometry, but rather on coarse topological data and degree mismatches.  In order to keep this paper brief we shall only cursorily review the operadic machinery~\cite{LodayVallette:AO}, definitions and properties of $\hyc$-algebras~\cite{Getzler:OMSGZRS}, $\bv$-algebras~\cite{Getzler:BVATDTFT}, or their homotopy versions~\cite{Getzler:OMSGZRS,GalvezCarilloTonksVallette:HBVA}, or K\"ahler geometry~\cite{GriffithsHarris:PAG}.  Linear maps, operad, and algebraic structures are all taken in a differential graded context over the complex numbers.

\section{Transfer of homotopy $\bv$-algebra structures and formal formality of the unit}
In this section, $\Oo$ shall be an operad, and $B$ shall be an $\Oo$-algebra with differential $d$ and cohomology $H$.

Loday-Vallette~\cite{LodayVallette:AO} is a good reference for the details of the general theory behind this section, which are greatly abbreviated.
\begin{definition}
Given any (possibly inhomogeneous) quadratic presentation of $\Oo$ one can produce a {\em Koszul dual} cooperad $\Oo^\ash$ cogenerated by the shifted generators of $\Oo$.  The cooperad $\Oo^\ash$ is equipped with a canonical map of $\sS$-modules $\kappa:\Oo^\ash\to \Oo$ which projects onto the cogenerators $\Oo^\ash$ and then includes them isomorphically as the space of generators of $\Oo$.

We call a presentation {\em good} if the map of operads $\cobar\Oo^\ash\to \Oo$ induced by $\kappa$ induces an isomorphism on homology, where $\cobar$ is the cobar functor.
\end{definition}
\begin{definition}
Let $\Oo$ be equipped with a good presentation.  A {\em homotopy $\Oo$-algebra} is an algebra over the operad $\cobar \Oo^\ash$.
\end{definition}
\begin{definition}
We call an element of $\Oo^\ash$ a {\em higher homotopy operation} if it is in the kernel of the map $\cobar\Oo^\ash\to\Oo$.
\end{definition}
\begin{definition}
{\em Linear transfer data} for $B$ consists of {\em linear} maps $\iota:H\to B$ and $\pi:B\to H$, along with a chain homotopy $h:B\to B$ so that $\pi\iota=\id_H$ and $dh+hd=\id_B-\iota\pi$

We say that linear transfer data $\iota$, $\pi$, and $h$ {\em satisfy the side conditions} if $h\iota=h^2=\pi h=0$.  
\end{definition}
\begin{construction}
Let $(\iota,\pi,h)$ be linear transfer data for $B$.  Suppose we are given an $\sS$-module $X$ and a map of $\sS$-modules $\kappa:X\to \Oo$.  

Then for any tree $T$ with $k$ leaves and vertices decorated by $X$, we can define an operation $H^{\otimes k}\to H$.  Redecorate $T$ as follows: decorate the leaves of $T$ with $\iota$, the internal edges by $h$, apply $\kappa$ to the vertex decorations, and decorate the root with $\pi$.  Then following the tree from leaves to root gives a recipe for composition of operations, most of which stays in $B$, using the $\Oo$-algebra structure and the transfer data.
\end{construction}
This construction is the standard way to define the transfered homotopy $\Oo$-algebra structure on $H$.
\begin{proposition}Let $\Oo$ be equipped with a good presentation, and let $(\iota,\pi,h)$ be linear transfer data for $B$.  There is a {\em transfered} homotopy $\Oo$-algebra structure on $H$ which extends the $\Oo$-algebra structure induced by that on $B$, and so that $\iota$ is the first term of a homotopy $\Oo$ quasi-isomorphism from $H$ to $B$.  Each generating operation in the transfered homotopy $\Oo$-algebra structure can be written as a sum of operations on $H$ induced by trees decorated with cogenerators of the cooperad $\Oo^\ash$ as in the construction.
\end{proposition}
\begin{proposition}
Suppose $\Oo$ has no internal differential.  Then the homotopy $\Oo$-algebra structure on $H$ can be truncated to the {\em strict} $\Oo$-algebra structure on $H$ induced by that on $B$ by forgetting all higher homotopy operations and arbitrarily choosing representatives in $\cobar\Oo^\ash$ for operations in $\Oo$.
\end{proposition}
Galvez-Carillo-Tonks-Vallette described a good presentation of the $\bv$ operad explicitly~\cite{GalvezCarilloTonksVallette:HBVA}.  For our purposes, the following facts suffice:
\begin{proposition}
The Koszul dual $\bv^\ash$ is cogenerated by the shifted duals of the binary product and bracket along with the unary $\Delta$ operator.  The Koszul dual $\bv^\ash$ can be graded by the number of $\Delta$ appearing in an element, and the differential reduces this grading by one.  

Call an element Lie type if every vertex is decorated by a bracket.  Every Lie type element is a coboundary.
\end{proposition}
\begin{definition}
In a homotopy $\bv$-algebra, we say that $e$ is a formal unit if $e\cdot x=x$ and any operation in the orthogonal complement of the product and the identity gives $0$ if $e$ is an argument.
\end{definition}
\begin{lemma}\label{lemma:formalunit}
Suppose we are given linear transfer data satisfying the side conditions between a $\bv$-algebra $B$ and $H$, and that $B$ has a unit $e$ so that $\iota([e])=e$.  Then $[e]$ is a formal unit in $H$.
\end{lemma}
\begin{proof}
Any transfered operation with an argument of $[e]$ must either take the bracket of $e$ with something, which is always $0$, or multiply by $e$, which is the identity.  If the tree governing this operation has any internal edges, then this results in either $hh$, $h\iota$, or $\pi h$ in the composed operation.
\end{proof}

\section{Interpreting a trivialized homotopy $\bv$-algebra as a homotopy $\hyc$-algebra}
Getzler~\cite{Getzler:OMSGZRS} described a good presentation of the $\hyc$ operad. The following proposition connects the resulting Koszul dual cooperad $\hyc^\ash$ to the cooperad $\bv^\ash$ described by Galvez-Carillo-Tonks-Vallette.
\begin{proposition}
There is a quotient of cooperads $\bv^\ash\to \hyc^\ash$, given by projecting onto the $\Delta^0$ grading and then quotienting further by the image of the differential.  
\end{proposition} 
\begin{proof}
This is essentially proven as Theorem~2.1 and Proposition~2.20 of~\cite{DrummondColeVallette:MMBVO}.  The only missing step follows from the fact that the differential of $\bv^\ash$ reduces the $\Delta$ grading by one while the cooperad structure map preserves the $\Delta$ grading.
\end{proof}
\begin{definition}
A {\em strongly trivialized} homotopy $\bv$ structure is one where any generator containing a vertex decorated by $\Delta$ acts as zero.
\end{definition}
\begin{lemma}
A strongly trivialized homotopy $\bv$-structure on a cochain complex induces a homotopy $\hyc$-structure on that complex.
\end{lemma}
\begin{proof}
The space of operations remaining after eliminating all operations with $\Delta$ is parameterized by the $\Delta^0$ part of the grading.  However, any boundary is the boundary of a generator that act as zero, and so acts as zero itself.  The terms in the differential of a homotopy $\bv$ operation are the sum of an operation that acts as zero and the differential of the same operation considered as a homotopy $\hyc$ operation.
\end{proof}
We can choose whatever representatives we like for the generators of the homotopy $\hyc$ operad.  For ease, we assume that we have chosen representative trees with no $\Delta$ anywhere, so our representatives are sums of trivalent trees with vertices decorated by products and brackets. We further assume that we have chosen a representative which is homogeneous in the number of brackets. 
\begin{definition}
We say that a homotopy $\hyc$-algebra is {\em obviously formal} if any generator containing more than one vertex decorated by the product acts as zero, and {\em formal} if it is isomorphic to an obviously formal structure.
\end{definition}
We shall describe a $\bv$-algebra and choose transfer data which strongly trivializes the transfered homotopy $\bv$-algebra on its cohomology in such a way that the induced homotopy $\hyc$-algebra is obviously formal.

\section{The Calabi-Yau example}\label{sec:CY}
We recall the Barannikov-Kontsevich construction~\cite{BarannikovKontsevich:FMFLAPV}. The transfer data is discussed by Merkulov~\cite{Merkulov:SHAKM}, and all of the commutation, adjunction, and other equalities pertaining to K\"ahler manifolds that we shall use can be found in Chapter 0 of Griffiths and Harris~\cite{GriffithsHarris:PAG}. 

For the remainder of the paper, $M$ shall be a Calabi-Yau variety of complex dimension $n$ and nonvanishing holomorphic volume form $\Omega$. 

We shall use the notation $B$ to refer to the bigraded complex
\[\Gamma(\wedge \overline{T}^*M \otimes \wedge TM)\]
with differential $\bar{\del}\otimes\id$.  $H$ shall denote the cohomology of $B$.

The map $\wedge TM\to \wedge T^{*}M$ given by contracting with $\Omega$ induces an isomorphism of cochain complexes between $B$ and the complex forms on $M$ with differential $\bar{\del}$, since in coordinates, $\Omega$ is nonvanishing and holomorphic.  

The complex $B$ is a commutative associative algebra, but contraction by $\Omega$ does not preserve this structure (the product on forms even has different degree).  The map $\del$ on $\Omega^* M$ induces a $\bv$ operator $\Delta$ on $B$, giving $B$ the structure of a $\bv$-algebra.  If we use the notation $\Star$ to denote both contracting by $\Omega$ and its inverse, then $\Delta=\Star \del \Star$.

Since $M$ is compact K\"ahler, the cohomology of this cochain complex coincides with its de Rham cohomology and the harmonic forms give $\bar{\del}$-representatives.  This gives us $\iota$ of our transfer data.  We shall work primarily in forms, suppressing $\Star$.  The product will be denoted:
\[\mu \wtm \nu = \Star(\Star \mu \wedge \Star \nu).\]
The projection $\pi$ is given in terms of the K\"ahler metric as follows: choosing an orthonormal basis $e_i$ of the harmonic forms, $\pi(\omega)=a_i [e_i]$, where
\[a_i=\int_M \omega\wedge \star e_i\]
where $\star \ne \Star$ is the Hodge star induced by the K\"ahler metric. The homotopy is of the form $\bar{\del}^* G=G\bar{\del}^*$ where $G$ is the inverse of the Laplacian on the orthogonal complement of the harmonic forms.  
\begin{lemma}\label{lemma:CYsideconditions}
$(\iota,\pi,h)$ constitute linear transfer data which satisfy the side conditions.
\end{lemma}
\begin{proof}
The global inner product $(x,y)$ of forms is $\int_M x\wedge \star y$ so $\pi \iota$ is the standard decomposition in terms of the basis $[e_i]$.  $\bar{\del}h+h\bar{\del}=G\Delta$, which is $0$ on the harmonic forms and the identity on their orthogonal complement, as desired.  The composition $h\iota$ is zero because harmonic forms are $\bar{\del}^*$ closed, $hh$ is zero because $(\bar{\del}^*)^2=0$, and $\pi h$ is zero because up to sign and conjugation, \[\int_M \star \bar{\del}\star x \wedge \star e_i = (\star e_i, \bar{\del}\star x)=(\star \bar{\del} e_i,\star x)=0\]
since $e_i$ is $\bar{\del}$ closed.
\end{proof}

\begin{lemma}
The transfered homotopy-$\bv$-algebra on $H$ is strongly trivialized.
\end{lemma}
\begin{proof}
If a vertex in a tree is decorated by $\Delta$, it is sandwiched between two edges which are decorated with $\iota$, $\pi$, or $h$.  The composition $\del \iota=0$ because harmonic forms are $\del$-closed.  The composition $\pi \del=0$ because \[\int_M \del x \wedge \star e_i = \int_M \del (x\wedge \star e_i) \pm \int_M x\wedge \star \del^* e_i;\]
the first term vanishes by Stokes' Theorem and the second because harmonic forms are $\del^*$-closed.  The remaining possibility is for both edges to be decorated with $h$, but the composition $h\del h$ is equal to
\[G\bar{\del}^* \del \bar{\del}^* G\]
which is zero since $\del$ and $\bar{\del}^*$ commute and square to zero.
\end{proof}
\begin{corollary}
The transfered homotopy $\bv$-algebra on $H$ induces a quotient homotopy $\hyc$-algebra on $H$.
\end{corollary}
%
\section{Universal formality at the top and bottom}\label{sec:topandbottom}
\begin{lemma}
The bidegree $(n,n)$ (resp. $(0,0)$) cohomology of $(B,\bar{\del})$ is one dimensional, represented by $\Star \overline{\Omega}$ (resp. $1=\Star\Omega$).
\end{lemma}
\begin{proof}
The proofs are the same up to conjugation.  To see the second statement, note that any $(0,0)$ element of cohomology would correspond to a harmonic form of bidegree $(n,0)$ under $\Star$, which would be a holomorphic functional coefficient times the holomorphic volume form.  That holomorphic functional coefficient is globally defined on the compact manifold $M$, therefore constant.
\end{proof}
\begin{lemma}\label{lemma:formalunittwo}
The class of the unit of $B$ is a formal unit in the transfered homotopy $\bv$-algebra.
\end{lemma}
\begin{proof}
This is a direct corollary of Lemmas~\ref{lemma:formalunit}~and~\ref{lemma:CYsideconditions}.
\end{proof}
So the unit is only involved in the product.  There is a dual statement for the $(n,n)$ cohomology.
\begin{proposition}\label{prop:formaltop}
The product is the only non-zero operation with output in bidegree $(n,n)$.
\end{proposition}
The presentation in terms of operadic algebra obscures the fact that this is dual to the unit statement, and makes it slightly messier to show.  We shall need the following lemma:
\begin{lemma}\label{lemma:easywedge}
Let $\mu$ and $\nu$ be homogeneous forms in complementary bidegree.  Then up to sign,
\[(\mu\wtm\nu)\wedge \Omega)=\mu\wedge\nu\]
\end{lemma}
\begin{proof}
Write $\mu\in\Omega^{i,j}M$ in coordinates as $\mu=\sum \mu_{IJ} d\bar{z}_I dz_J$ (likewise for $\nu$) and write $\Omega=\eta dz_{all}$  Ignoring signs,
\[\Star\mu\wedge \Star\nu=\sum_{I,J}\frac{\mu_{IJ}\nu_{\hat{I}\hat{J}}}{\eta^2}d\bar{z}_{all}\del_{z,all}\]
Applying $\Star$ takes us to
\[\sum_{I,J}\frac{\mu_{IJ}\nu_{\hat{I}\hat{J}}}{\eta}d\bar{z}_{all}\]
and so multiplying by $\Omega$ gives
\[\sum_{I,J}\mu_{IJ}\nu_{\hat{I}\hat{J}}d\bar{z}_{all}dz_{all}=\mu\wedge\nu.\]
\end{proof}

\begin{proof}[Proof of Proposition~\ref{prop:formaltop}]
Any transfered operation is a sum of decorated trees, each of which ends with a product or a bracket.  The inputs to the final operation in one of the transfer formula summands are either leaves of the tree, which come from $\iota$ as harmonic representatives, or in the image of the homotopy $h$, and for any operation other than the product, in particular any higher homotopy operation, at least one of them has to be in the image of $h$.  Specifically, both arguments to the last operation are $\bar{\del}^*$-closed, and at least one is $\bar{\del}^*$-exact. 

We shall show that each of these summands vanishes in the case where the output lands in bidegree $(n,n)$.  


The projection $\pi$, in bidegree $(n,n)$ in $B$, viewed as degree $(0,n)$ in $\Omega^*M$, is:
\[\pi(\sigma) = [\overline{\Omega}]\int_M \sigma\wedge \star\overline{\Omega}=\pm [\overline{\Omega}]\int_M \sigma \wedge \Omega\]

Assume first that the summand operation we are analyzing ends in a product $\pi (\mu\wtm\nu)$.  

We may assume that $\mu$ is $\bar{\del}^*$ closed and $\nu$ is a $\bar{\del}^*$-boundary.  Then, using Lemma~\ref{lemma:easywedge}:
\[
\int_M (\mu\wtm\bar{\del}^* \lambda)\wedge\Omega=
\int_M \mu\wedge\bar{\del}^* \lambda =
(\mu, \bar{\del}\star \lambda)=
(\bar{\del}^* \mu,\star \lambda)=0.\]

Now suppose instead that the summand ends in a bracket $[\mu,\nu]$.  The coefficient of $[\overline{\Omega}]$ is the sum of three terms:
\[\int_M \del (\mu\wtm\nu)\wedge \Omega - \int_M (\del \mu \wtm \nu)\wedge \Omega \pm \int_M (\mu\wtm\del \nu)\wedge \Omega.\]
The first term vanishes by Stokes' Theorem and the other two by the same argument as for the product, since $\del$ and $\bar{\del}^*$ commute.
\end{proof}

\section{Formal formality in low dimension}
The final ingredient we shall need to specialize is the nonzero entries of the Hodge diamond.  
\begin{definition}
A bigraded complex has {\em hypersurface footprint} (of dimension $n$) if it linearly decomposes into a {\em formal part} in bidegree $(i,n-i)$ and a {\em primitive part} in bidegree $(j,j)$, and if the $(0,0)$ and $(n,n)$ bidegree components are one dimensional.\footnote{These bidegrees are reversed from the expected bidegrees in forms.  Also, {\em primitive} is just a convenient shorthand and does not refer to the Lefschetz decomposition.} We say $M$ has hypersurface footprint if the cohomology of $B$ does.
\end{definition}
\begin{lemma}
Every Calabi-Yau variety of dimension $2$ or $3$ has hypersurface footprint.
\end{lemma}
\begin{proof}
By definition, $H^{1,0}=H^{0,1}=0$. Poincar\'e and Serre duality imply that all other bigraded homology groups outside the hypersurface footprint are zero in these two dimensions.  
\end{proof}

We shall count degrees in order to prove the main theorem. Consider applying transfered homotopy $\hyc$ operations to the cohomology $H$.  
\begin{lemma}
Any higher homotopy $\hyc$ operation $H^{\otimes k}\to H$ on a transfered Calabi-Yau $\bv$-algebra is bidegree $(-\ell,-k+2)$ with $0\le \ell<k-2$.  The $\hyc$ operations in the strict truncation are of bidegree $(-k+2,-k+2)$.
\end{lemma}
\begin{proof}
An operation is calculated as a sum of trees decorated with products and brackets.  The Lie-type operations all act as zero.  The operations with one vertex decorated with the product are precisely the strict $\hyc$ operations, and all higher homotopy operations come from trees with at least two vertices decorated by products.  Each internal edge in a tree has bidegree $(0,-1)$ and each bracket has bidegree $(-1,0)$.  The total number of internal edges is $k-2$ and the total number of vertices is $k-1$.
\end{proof}
Now assume $M$ has hypersurface footprint.
\begin{lemma}
In any dimension, every higher homotopy operation from the primitive part to the primitive part acts as zero.
\end{lemma}
\begin{proof}
A higher homotopy operation would leave the diagonal since $\ell<k-2$.
\end{proof}
\begin{lemma}
In dimension four or less, every higher homotopy operation from the primitive part to the formal part acts as zero.
\end{lemma}
\begin{proof}
Each argument has degree at least two.  The minimal degree of a higher homotopy operation is $-2k+5$ and the minimal total degree of the inputs is $2k$ so the minimal degree of the output is $5$.
\end{proof}
\begin{lemma}
In dimension four or less, every higher homotopy operation with a single formal argument acts as zero.
\end{lemma}
\begin{proof}
The minimal degree of a higher homotopy operation is $-2k+5$ and the minimal total degree of the inputs is $2k+n-2$, so the minimal degree of the output is $n+3$.  For $n<5$, $n+3$ is greater than $2n-2$, which is the top degree possible for the target of a higher homotopy operation.
\end{proof}
\begin{lemma}
In any dimension above one, the product is the only non-zero operation with two or more formal arguments.
\end{lemma}
\begin{proof}
The minimal total degree of the inputs is $2n+2k-4$, and the minimal degree of the output is $2n$, which only accepts the product.
\end{proof}
The preceding four lemmas suffice to prove the main theorem.
\begin{remark} In the Calabi-Yau example, the hypercommutative algebra structure on the homology of $B$ can also be viewed as a linear pencil of torsion free flat connections on the moduli space of B-models. The hypercommutative structure constitutes the main part of a Frobenius manifold.  The homotopy hypercommutative structure is a kind of pencil of superconnections.  The main result, stated in this language, implies that for the B-model with a Calabi-Yau threefold target, the higher parts of the superconnections all vanish, leaving a pencil of ordinary connections~\cite{Park:SCQFTFM}.
\end{remark}
\bibliography{references-1}{}
\bibliographystyle{ieeetr}
\end{document}